\documentclass[12pt]{amsart}
\usepackage[
backend=biber,
style=numeric,
sorting=nyt
]{biblatex}
\DeclareNameAlias{default}{family-given}
\usepackage{epsfig,color}
\usepackage{blindtext}
\usepackage{graphicx}
\usepackage{enumitem}
\usepackage{url}
\usepackage{amssymb}
\usepackage{graphicx,import}
\usepackage{comment}
\usepackage{esint}
\usepackage{xcolor}
\usepackage{mathtools}
\usepackage{comment}
\usepackage{tikz}

\usepackage[toc,page]{appendix}

\allowdisplaybreaks

\usepackage[margin = 0.75in] {geometry}

\usepackage{hyperref}
\usepackage{dsfont}

\newtheorem{theorem}{Theorem}[section]

\newtheorem{proposition}[theorem]{Proposition}
\newtheorem{lemma}[theorem]{Lemma}
\newtheorem{corollary}[theorem]{Corollary}

\theoremstyle{definition}

\newtheorem{definition}[theorem]{Definition}
\newtheorem*{remark*}{Remark}

\newtheorem{remark}[theorem]{Remark}

\numberwithin{equation}{section}

\title{A Penrose-Type Inequality for Static Spacetimes}

\author{Brian Harvie}
\address{University of Copenhagen, Copenhagen, DE}
\email{brdh@math.ku.dk}

\begin{document}

\begin{abstract}
We establish a lower bound on the total mass of the time slices of $(n+1)$-dimensional asymptotically flat standard static spacetimes under the timelike convergence condition (TCC). The inequality can be viewed equivalently as a Minkowski-type inequality in these spaces, i.e. as a lower bound on the total mean curvature of the boundary, and thus extends inequalities from \cite{BHW12}, \cite{W18}, \cite{M18}, and \cite{HW24a}. Equality is achieved only by slices of Schwarzschild space and is related to the characterization of quasi-spherical static vacuum metrics from \cite{HW24a}. As a notable special case, we obtain the Riemannian Penrose inequality in all dimensions for static spaces under the TCC.
\end{abstract}
\maketitle
\section{Introduction}
Among the most important solutions of the Einstein field equations are ones which admit a time-translation symmetry, known as static spacetimes. A  \textit{standard static spacetime} is a Lorentzian manifold $(L^{n+1}, \widehat{g})$ of the form

\begin{equation} \label{lorentzian}
    L^{n+1} = M^{n} \times \mathbb{R} \hspace{1cm} \widehat{g} = - V(x)^{2} d\tau^{2} + g, \hspace{1cm} x \in M^{n},
\end{equation}
where $M^{n}$ is a smooth manifold, $g$ is a Riemannian metric on $M^{n}$, and $V$ is a positive smooth function on $M^{n}$. A standard static spacetime satisfies the \textit{timelike convergence condition} (TCC) if

\begin{equation} \label{dec}
  \widehat{\text{Ric}} (X,X) \geq 0 \hspace{1cm} \text{for timelike vectors    } X.
\end{equation}
Geometrically, the TCC, which is equivalent to the strong energy condition (SEC) in general relativity, implies that nearby timelike geodesics in $(L^{n+1},\widehat{g})$ tend to converge compared to those in flat spacetime. For a comprehensive review of matter sources for the Einstein equations satisfying this condition, see \cite{EC16}.

Throughout this note, we identify the spacetime \eqref{lorentzian} with the base manifold $(M^{n},g)$ equipped with the lapse function $V$ and refer to $(M^{n},g,V)$ as a \textit{static triple}. Our main result is a lower bound on the mass of asymptotically flat static triples under \eqref{dec}. Here, our definition of asymptotic flatness (Definition \ref{af}) allows for decay order $o_{2}(1)$ and is therefore insufficient to define the ADM mass of $(M^{n},g)$. Nevertheless, the class of static manifolds considered admit a notion of total mass $m$ via the expansion \eqref{V_decay} which coincides with the ADM mass whenever well-defined. 

\begin{theorem}[Penrose-Type Inequality for Static Spaces] \label{main}
Let $(M^{n},g,V)$ be a static triple which is asymptotically flat in the sense of Definition \ref{af} and which satisfies the TCC \eqref{dec}. Let $\Sigma^{n-1} \subset \partial M^{n}$ be an outer-minimizing boundary component of $M$. Assume that
\begin{itemize}
    \item $n > 7$ and $\partial M^{n}=\Sigma^{n-1}$ is connected, OR
    \item $3 \leq n \leq 7$ and $\partial M^{n} \setminus \Sigma^{n-1}$, if non-empty, consists of horizons, i.e. closed locally area-minimizing hypersurfaces. 
\end{itemize}
Then we have the following lower bound on the mass $m$ of $(M^{n},g)$:
 
\begin{equation} \label{main_ineq}
    m \geq \frac{1}{2} \left(\frac{|\Sigma|}{w_{n-1}} \right)^{\frac{n-2}{n-1}} - \frac{1}{2(n-1)w_{n-1}} \int_{\Sigma} VH d\sigma,
\end{equation}
where $w_{n-1}$ is the area of the unit sphere and $|\Sigma|$ is the $(n-1)$-dimensional area of $\Sigma$. Moreover, equality holds in \eqref{main_ineq} if and only if $(M^{n},g,V)$ is isometric to a rotationally symmetric region $\{ r \geq r_{0} \}$ of Schwarzschild space with mass $m$,

\begin{equation} \label{schwarzschild}
    M^{n}= \mathbb{S}^{n-1} \times (r_{m},\infty), \hspace{0.5cm} g_{m} = \left( 1 - \frac{2m}{r^{n-2}} \right)^{-\frac{1}{2}} dr^{2} + r^{2} g_{S^{n-1}}, \hspace{0.5cm} V_{m}= \sqrt{1 - \frac{2m}{r}},
\end{equation}
$m \in \mathbb{R}$ (not neccessarily $\mathbb{R}^{+}$) and $r_{m}= \left( \max \{ 0, 2 m \}  \right)^{\frac{1}{n-2}}$ is the Schwarschild radius.
\end{theorem}
This lower bound is an extension of the Minkowski-type inequality that was discovered by Brendle-Hung-Wang \cite{BHW12} in anti de Sitter Schwarzschild space and later extended to time slices of static vacuum spacetimes with zero cosmological constant by Wei \cite{W18}, McCormick \cite{M18}, and Harvie-Wang \cite{HW24a}. Theorem \ref{main} replaces the vacuum assumption with the TCC and recovers the same rigidity statement as in \cite{HW24a}.

A notable special case of inequality \eqref{main_ineq} occurs when $\Sigma$ is a horizon.

\begin{corollary} [Riemannian Penrose Inequality for Static Spaces] \label{riem_penrose}
Let $(M^{n},g,V)$ be an asymptotically flat static triple satisfying the TCC \eqref{dec} and with boundary $\partial M = \cup_{i=1}^{I}\Sigma_{i}$ consisting of outer-minimizing horizons $\Sigma_{i}$. If $n > 7$, then suppose additionally that $I=1$, i.e. that $\partial M=\Sigma_{1}$ is connected.
Then the Riemannian Penrose inequality

\begin{eqnarray}
    m &\geq& \frac{1}{2} \left( \frac{|\Sigma_{i}|}{w_{n-1}} \right)^{\frac{n-2}{n-1}}, \label{penrose}
\end{eqnarray}
holds for each horizon $\Sigma_{i}$, where $m$ is the mass of $(M^{n},g,V)$ defined in \eqref{mass_bound}. Furthermore, equality holds in \eqref{penrose} if and only if $I=1$ and $(M^{n},g)$ is isometric to Schwarzschild space of mass $m$.
\end{corollary}
\begin{remark}
When $(M^{n},g)$ is a complete Riemannian manifold, taking $\Sigma^{n-1}$ in \eqref{main_ineq} to be a geodesic sphere of arbitrarily small radius also yields the Riemannian positive mass theorem-- that is, $m \geq 0$ for a complete asymptotically flat static triple $(M^{n},g,V)$ under the TCC. However, the asymptotic profile \eqref{V_decay} for $V$ also yields non-negative mass when $(M^{n},g)$ is complete via a maximum principle argument. Because of this, we only mention this property in passing.
\end{remark}
Agostiniani-Mazzieri \cite{AM15} proved \eqref{penrose} for asymptotically flat static vacuum triples, and Agostiniani-Mazzieri-Oronzio \cite{AMO20} later proved a capacity version of \eqref{penrose} for certain asymptotically flat static triples $(M^{n},g,V)$ with multiple ends. The assumptions in \cite{AMO20} are equivalent to the TCC \eqref{dec} and the additional condition that $\widehat{Ric}(\partial_{\tau},\partial_{\tau}) \equiv 0$ in the corresponding spacetime $(L^{n+1},\widehat{g})$, as well as the more standard decay order for the metric $g_{ij}$. 
 
The outline of the paper is as follows: in Section 2, we examine the geometric constraints that the timelike convergence condition imposes on the triple $(M^{n},g,V)$. Namely, the lapse function $V$ is sub-harmonic, and the $2$-tensor $P$ defined in \eqref{P} is non-negative definite. We then derive a Smarr-type formula for the mass in terms of a flux integral of $V$.

Inequality \eqref{main_ineq} is proven using inverse mean curvature flow (IMCF). Given the results of Section 2, we show within an asymptotically flat static triple $(M^{n},g,V)$ satisfying the TCC that the quantity 

\begin{equation*}
    Q(t)= \left( |\Sigma_{t}| \right)^{-\frac{n-2}{n-1}} \left( \int_{\Sigma_{t}} VH d\sigma + 2(n-1)w_{n-1} m \right)
\end{equation*}
is monotonically non-increasing under IMCF. This is the same functional considered in \cite{W18} and \cite{M18}. In Section 3, we show monotonicity for the smooth flow, and in Section 4 we generalize this to the weak flow. To preserve weak flow monotonicity when $3 \leq n \leq 7$ and $\partial M^{n} \setminus \Sigma^{n-1}$ contains additional horizon boundary components, we require Lemma \ref{horizon}, which states that the lapse function vanishes on horizons. This is a straightforward extension of the main result in \cite{HMM17}, and so we include the proof of Lemma \ref{horizon} in the appendix. We also need to employ the jumping procedure in weak IMCF in the presence of these components-- this procedure is described thoroughly in \cite{HI99} and \cite{W18}.

We establish rigidity of \eqref{main_ineq} in Section 5. The proof builds upon earlier joint work \cite{HW24}, \cite{HW24a} with Ye-Kai Wang that relates inequality \eqref{main_ineq} to quasi-spherical metrics. Quasi-spherical metrics were originally considered by Bartnik \cite{B93} and Shi-Tam \cite{ST02} for their role in quasi-local mass. In \cite{HW24a}, we showed that a static vacuum metric on $\mathbb{S}^{n-1} \times (r_{0},\infty)$ of the form $u^{2}(\theta,r) dr^{2} +r^{2} g_{S^{n-1}}$ belongs to the Schwarzschild family. This implies rigidity in the vacuum case as a by-product. Under the TCC, the monotonicity formula for $Q(t)$ guarantees that the metric is of this form whenever equality is achieved, at least outside of a compact set. However, the vacuum condition is not obvious because the monotonicity formula is independent of scalar curvature. To show that $(M^{n},g)$ saturating \eqref{main_ineq} are indeed static vacuum, we combine the Bianchi and Ricci identities in the quasi-spherical gauge and invoke the scalar curvature decay for the metric $g$. We also use the fact that weak IMCF becomes smooth in asymptotically flat backgrounds, Theorem 2.1 in \cite{HW24a}.
 
\subsection*{Acknowledgements} I would like to thank Professor Ye-Kai Wang for a very helpful suggestion about the rigidity problem.
\section{The timelike convergence condition and the ADM Mass}
We will denote by $\text{Ric}$, $\nabla^{2} V$, and $\Delta V$ the Ricci tensor, Hessian tensor of $V$, and Laplacian of $V$ on the base manifold $(M^{n},g)$. These are related to the ambient Ricci curvature at a point $(x,\tau) \in (L^{n+1},\widehat{g})$ by

\begin{eqnarray}
\widehat{\text{Ric}} (E_{0},E_{0}) &=& \frac{1}{V} \Delta V, \label{nn} \\
\widehat{\text{Ric}} (E_{0},E_{i}) &=& 0, \hspace{5cm} i,j \in \{ 1, \dots, n \}, \label{nt} \\
\widehat{\text{Ric}} (E_{i}, E_{j}) &=& \text{Ric} (E_{i},E_{j}) - \frac{1}{V} \nabla^{2} V(E_{i},E_{j}), \label{tt}
\end{eqnarray}
where
\begin{eqnarray*}
    E_{0}&=& \frac{1}{V(x)} \frac{\partial}{\partial \tau} \in T_{(x,\tau)} L^{n+1}, \hspace{4cm} \text{and} \\
     E_{1},\dots, E_{n}  & & \text{  is an O.N.B. of      } \hspace{0.2cm} T_{(x,\tau)} (M^{n} \times \{ \tau \}) \subset T_{(x,\tau)} L^{n+1}.
\end{eqnarray*}
We omit the computation of \eqref{nn}-\eqref{tt} because it is elementary and refer the reader e.g. to \cite{DU09}, Proposition 2.3 for the calculations. First, we discuss local geometric conditions that these equations impose on the triple $(M^{n},g,V)$.
\begin{proposition} [The TCC on Time Slices] \label{TCC_slice}
Under the timelike convergence condition \eqref{dec}, the lapse function $V$ is sub-harmonic on $(M^{n},g)$, i.e. 

\begin{equation} \label{subharmonic}
    \Delta V \geq 0,
\end{equation}
and the symmetric $2$-tensor $P$ on $M^{n}$ defined by 

\begin{equation} \label{P}
    P= V \text{Ric} - \nabla^{2} V + \left( \Delta V \right) g
\end{equation}
is non-negative definite.
\end{proposition}

\begin{proof}
Sub-harmonicity of $V$ is clear from equation \eqref{nn}. To see \eqref{P}, let $E_{1},\dots,E_{n}$ be an orthonormal basis of $T_{x} M^{n}$ which diagonalizes $P$. Then for each $i \in \{ 1,\dots,n\}$ the vector 

\begin{equation*}
    N_{i} = E_{0} + E_{i} \in T_{(x,\tau)} L^{n+1},
\end{equation*}
is null, i.e. $\widehat{g}(N_{i},N_{i})=0$. Therefore, the timelike convergence condition \eqref{dec} implies that $\widehat{Ric}(N_{i},N_{i}) \geq 0$ by continuity, and so equations \eqref{nn}-\eqref{tt} yield

\begin{eqnarray*}
    0 &\leq&  \widehat{\text{Ric}} (N_{i}, N_{i})\\
    &=& \widehat{\text{Ric}} (E_{0},E_{0}) + \widehat{\text{Ric}} (E_{i},E_{i}) + 2 \widehat{\text{Ric}} (E_{0},E_{i}) \\
    &=& \text{Ric}(E_{i},E_{i}) - \frac{1}{V} \nabla^{2} V(E_{i},E_{i}) + \frac{1}{V} \Delta V= \frac{1}{V}P(E_{i},E_{i}).
\end{eqnarray*}
\end{proof}
Next, we discuss the asymptotic conditions on the triple $(M^{n},g,V)$.
\begin{definition} \label{af}
A static triple $(M^{n},g,V)$ is \textit{asymptotically flat} if there exists a compact set $K$ and a diffeomorphism $x: \mathbb{R}^{n} \setminus B_{1}(0) \rightarrow M^{n} \setminus K$ such that in the coordinates $x=(x^{1},\dots,x^{n})$ we have the asymptotic expansions
\begin{eqnarray}
    g_{ij}(x) &=& \delta_{ij} + o_{2}(1), \label{g_decay} \\
    V(x) &=& 1 - m |x|^{2-n} + o_{1}(|x|^{2-n}) \label{V_decay}
\end{eqnarray}
for the metric and lapse, where $\delta_{ij}$ is the Euclidean metric on $\mathbb{R}^{n} \setminus B_{1}(0)$. We refer to the coefficient $m$ in \eqref{V_decay} as the \textit{mass} of $(M^{n},g)$.
\end{definition}
If $(M^{n},g)$ is asymptotically flat with decay order $q > \frac{n-2}{2}$ and the tensor $P_{ij}$ in \eqref{P} vanishes outside of some large compact set $K$, then the lapse function $V$ automatically has the expansion \eqref{V_decay} near infinity after scaling, c.f. \cite{HMM17} Proposition B4. In this situation, the coefficient $m$ in \eqref{V_decay} coincides with the ADM mass of $(M^{n},g)$. Here, we directly assume the expansion \eqref{V_decay} as in \cite{KL18} and \cite{YL15}. We also mention that this notion of asymptotic flatness is slightly different from the one in \cite{HW24a}, which does not assume a first-order expansion for the lapse. 

The expansion \eqref{V_decay} allows us to express the ADM mass of $(M^{n},g)$ as a flux integral at infinity of the lapse. This insight is credited originally to Smarr \cite{S73} and is central to our monotonicity argument.

\begin{proposition} \label{mass_formula_theorem}
Let $(M^{n},g,V)$ be asymptotically flat. Then the mass $m$ of $(M^{n},g)$ is given by the formula
\begin{equation} \label{mass_formula}
     m = \frac{1}{(n-2)w_{n-1}} \lim_{r \nearrow \infty} \int_{|x| = r} \frac{\partial V}{\partial \nu} d\sigma
\end{equation}
in the coordinates $(x^{1},\dots,x^{n})$ from Definition \eqref{af}, where $\nu$ is outward unit normal of the hypersurface $\{ |x| = r \}$.
\end{proposition}
\begin{proof}



Using the expansion \eqref{V_decay}, we find in polar coordinates that
\begin{eqnarray*}
    \frac{\partial V}{\partial r} &=& \frac{\partial}{\partial r} \left( 1 - m r^{2-n} + o_{1}(r^{2-n}) \right) = (n-2) m r^{1-n} + o(r^{1-n}). 
\end{eqnarray*}
Therefore,
\begin{eqnarray*}
    \int_{|x| = r } \frac{\partial V}{\partial \nu} d\sigma = \int_{|x|=r} \frac{\partial V}{\partial r} d\sigma + o(1)= (n-2)w_{n-1} m + o(1).
\end{eqnarray*}
\end{proof}

\section{A Monotone Quantity for the Inverse Mean Curvature Flow}
In this section, we introduce a quantity which we show to be monotonically non-increasing along inverse mean curvature flow (IMCF) in asymptotically flat static manifolds under the TCC. Recall that a one-parameter family of embeddings $X: \Sigma^{n-1} \times [0,T) \rightarrow (M^{n},g)$ of a smooth closed manifold $\Sigma^{n-1}$ into a Riemannian manifold $(M^{n},g)$ solves IMCF if
\begin{equation} \label{IMCF}
\frac{\partial X}{\partial t} (x,t) = \frac{1}{H} \nu(x,t),    \hspace{2cm} (x,t) \in \Sigma^{n-1} \times [0,T)
\end{equation}
where $\nu$ is the outward-pointing unit normal and $H>0$ the mean curvature of the hypersurface $\Sigma_{t}=X_{t}(\Sigma^{n-1}) \subset (M^{n},g)$. 

To obtain a monotone quantity along \eqref{IMCF}, we require an upper bound the flux of $V$ over the leaves $\Sigma_{t}$. For this step, we need the following lemma for the case of multiple boundary components.

\begin{lemma} \label{horizon}
Let $(M^{n},g,V)$ be asymptotically flat with boundary $\partial M$, and let $\Sigma_{i} \subset \partial M$ be a horizon. Suppose the tensor $P$ defined in \eqref{P} is non-negative definite. Then $V|_{\Sigma_{i}} \equiv 0$.
\end{lemma}
Note that non-negativity of $P$ is equivalent to the null convergence condition (NCC) in the spacetime $(L^{n+1},\widehat{g})$, which is the weakest (pointwise) energy condition in general relativity. Galloway proved a version of this in \cite{G93}, c.f. also Huang-Martin-Miao \cite{HMM17} for the case of a sign-changing lapse function. It is straightforward to generalize these proofs, so we will save this for the appendix.

By combining Lemma \ref{horizon} with the mass formula from Proposition \ref{mass_formula_theorem}, we obtain the desired bound on the flux integral on $\Sigma_{t}$.
\begin{proposition} \label{imcf_mass}
Let $(M^{n},g,V)$ be asymptotically flat with outer-minimizing boundary $\partial M$, and let $\Sigma \subset \partial M$ be a smooth boundary component of $\partial M$ with (outward) mean curvature $H > 0$. Suppose that the remaining boundary components $\Sigma_{i} \subset \partial M$, $i =1, \dots, I$ , are horizons, and consider IMCF $X: \Sigma^{n-1} \times (0,T) \rightarrow (M^{n},g)$ with initial condition $\Sigma_{0}=\Sigma$ . Then under the TCC, we have

\begin{equation} \label{mass_bound}
    \int_{\Sigma_{t}} \frac{\partial V}{\partial \nu} d\sigma \leq (n-2) w_{n-1} m
\end{equation}
for each $t \in (0,T)$.
\end{proposition}

\begin{proof}
We denote $\Omega_{t}= X\left( \Sigma^{n-1} \times (0,t) \right)$ (notice $\partial \Omega_{t} =\Sigma_{t} \cup \Sigma_{0}$). By the divergence theorem, sub-harmonicity of $V$, and the mass formula \eqref{mass_formula}, we have that

\begin{eqnarray}
    0 &\leq& \int_{M \setminus \Omega_{t}} \Delta V d\Omega = \lim_{r \nearrow \infty} \int_{|x| = r} \frac{\partial V}{\partial \nu} d\sigma - \int_{\Sigma_{t}} \frac{\partial V}{\partial \nu} d\sigma - \sum_{i=1}^{I} \int_{\Sigma_{i}} \frac{\partial V}{\partial \nu} d\sigma \label{divergence} \\
    &=& (n-2)w_{n-1} m - \int_{\Sigma_{t}} \frac{\partial V}{\partial \nu} d\sigma - \sum_{i=1}^{I} \int_{\Sigma_{i}} \frac{\partial V}{\partial \nu} d\sigma. \nonumber
\end{eqnarray}
Now, $V|_{\Sigma_{i}} \equiv 0$ for each $i = 1, \dots, I$ by Lemma \ref{horizon}. Since $V > 0$ in $M^{n}$, this means that $\frac{\partial V}{\partial \nu} \geq 0$ on each $\Sigma_{i}$. So altogether \eqref{divergence} yields \eqref{mass_bound}.
\end{proof}
The monotone quantity considered in this paper is the same as the one in \cite{W18} and \cite{M18}. The latter paper also requires the bound \eqref{mass_bound} for static vacuum triples. Here, examining the complete variation formula reveals the significance of the TCC.

\begin{theorem}
Let $(M^{n},g,V)$ be as in Proposition \ref{imcf_mass}. Then the quantity
\begin{equation} \label{Q}
    Q(t) = \left( |\Sigma_{t}| \right)^{-\frac{n-2}{n-1}} \left( \int_{\Sigma_{t}} VH d\sigma + 2(n-1) w_{n-1} m \right)
\end{equation}
is monotonically non-increasing under IMCF with initial condition $\Sigma$. 
\end{theorem}
\begin{proof}
By the variation formula for the area element and mean curvature, see for example \cite{HP99}, we have under \eqref{IMCF} that
\begin{eqnarray}
    \frac{\partial}{\partial t} \sqrt{|\sigma|} &=& \sqrt{|\sigma|}, \label{vol_element} \\
    \frac{\partial}{\partial t} H &=& -\Delta_{\Sigma_{t}} \frac{1}{H} - (|A|^{2} + \text{Ric}(\nu,\nu)) \frac{1}{H}.
\end{eqnarray}
From this, we compute

\begin{eqnarray}
    \frac{d}{dt}  \int_{\Sigma_{t}} VH d\sigma &=& \int_{\Sigma_{t}} V \frac{\partial H}{\partial t} +  H \frac{\partial V}{\partial t} + + VH d\sigma \nonumber \\
    &=& \int_{\Sigma_{t}}  V \left( -\Delta_{\Sigma_{t}} \frac{1}{H} - (|A|^{2} + \text{Ric}(\nu,\nu)) \frac{1}{H} \right) + \frac{\partial V}{\partial \nu} + VH d\sigma. \nonumber \\
    &=& \int_{\Sigma_{t}} - \frac{1}{H} \left( \Delta_{\Sigma_{t}} V + H \frac{\partial V}{\partial \nu} + V \text{Ric}(\nu,\nu) \right) + 2 \frac{\partial V}{\partial \nu} - \frac{|A|^{2}}{H} V + VH d\sigma \nonumber \\
    &=& \int_{\Sigma_{t}} - \frac{1}{H} \left( \Delta V - \nabla^{2} V(\nu,\nu) + V\text{Ric}(\nu,\nu)\right) + 2 \frac{\partial V}{\partial \nu} - |\mathring{A}|^{2} \frac{V}{H} + \frac{n-2}{n-1} VH d\sigma, \label{ev_Q}
\end{eqnarray}
where $\mathring{A}$ is the trace-free second fundamental form of $\Sigma_{t}$. In the third line, we applied the divergence theorem over $\Sigma_{t}$ to the Laplacian term, and in the last line we used the identity $|A|^{2}= \frac{1}{n-1}H^{2} + |\mathring{A}|^{2}$. The first term in \eqref{ev_Q} equals $-\frac{1}{H} P(\nu, \nu)$ and hence is non-positive by Proposition \ref{TCC_slice}, and the second term is bounded above by $2(n-2)w_{n-1} m$ according to \eqref{mass_bound}. Therefore, \eqref{ev_Q} implies
\begin{equation*}
    \frac{d}{dt}  \left( \int_{\Sigma_{t}} VH d\sigma + 2(n-1) w_{n-1} m \right) \leq \frac{n-2}{n-1} \left( \int_{\Sigma_{t}} VH d\sigma + 2(n-1) w_{n-1} m \right).
\end{equation*}
On the other hand, by the evolution \eqref{vol_element} of the area element we have

\begin{equation*}
    \frac{d}{dt} |\Sigma_{t}|^{- \frac{n-2}{n-1}} = - \frac{n-2}{n-1} |\Sigma_{t}|^{-\frac{n-2}{n-1}},
\end{equation*}
and so $\frac{d}{dt} Q(t) \leq 0$. 


\end{proof}


    
\section{Monotonicity for Weak Inverse Mean Curvature Flow}
To prove Theorem \ref{main}, we must prove that the quantity $Q(t)$ from the previous section remains monotone under weak inverse mean curvature flow. Huisken and Ilmanen introduced weak IMCF in \cite{HI99} to prove the Riemannian Penrose inequality, and since then this formulation has been used to prove many other geometric inequalities in general relativity, c.f. for example \cite{FS13}, \cite{LN13}, \cite{KWY17}, and \cite{C16}, among many others. A weak solution $\Sigma_{t}$ of IMCF in $(M^{n},g)$ is identified with a function $u \in C^{0,1}_{\text{loc}} (M^{n})$, with the flow surfaces $\Sigma_{t}$ corresponding to the level sets of $u$. We do not give the precise definition from \cite{HI99} here, but we will highlight the properties of the weak flow surfaces $\Sigma_{t}$ needed for our proof. In this section, $H$ denotes the weak mean curvature of $\Sigma_{t}$, defined in Section 1 of \cite{HI99} using variations of area.

\begin{theorem}[Properties of weak IMCF] \label{properties}
Let $(M^{n},g)$ be an asymptotically flat Riemannian manifold with connected smooth boundary $\partial M$, and let $u \in C^{0,1}_{\text{loc}}(M^{n})$ be the proper weak IMCF with initial condition $\Sigma_{0}=\partial M$. Then the following hold, with the corresponding reference from \cite{HI99} provided:

\begin{enumerate}[label= (\roman*)]
    \item (Theorem 1.3, Regularity) For each $t >0$, the sets $\Sigma_{t}= \partial \{ u < t \}$ and $\Sigma^{+}_{t}=\partial \{ u > t \}$ are outer-minimizing, resp. strictly outer-minimizing, $C^{1,\alpha}$ hypersurfaces for $3 \leq n \leq 7$. For $n > 7$, $C^{1,\alpha}$ regularity holds away from a singular set $Z$ of Hausdorff dimension at most $n-7$.
    \item ((1.12) and Lemma 5.1, Mean-convexity) For a.e. $t >0$, $\Sigma_{t}$ and $\Sigma'_{t}$ have essentially bounded weak mean curvature $H_{\Sigma_{t}} (x) = |\nabla u(x)| >0$ for $\mathcal{H}^{n-1}$ a.e. $x \in \Sigma_{t}$.
    \item ((1.10), Convergence) For each $t>0$, $\alpha \in (0,1)$, and $3 \leq n \leq 7$, we have the convergences

    \begin{equation} \label{convergence}
        \Sigma^{+}_{s} \rightarrow \Sigma_{t} \hspace{1cm} \text{as} \hspace{1cm} s \nearrow t, \hspace{2cm} \Sigma_{s} \rightarrow \Sigma^{+}_{t} \hspace{1cm} \text{as} \hspace{1cm} s \searrow t
    \end{equation}
    in $C^{1,\alpha}$. For $n > 7$, the convergence is true away from a singular set $Z$ of Hausdorff dimension less than $n-7$.

    \item (Lemma 5.6, Area Growth) The $(n-1)$-Hausdorff measure of $\Sigma_{t}$ and $\Sigma^{+}_{t}$ is given by

    \begin{equation} \label{area_growth}
        |\Sigma_{t}| = |\Sigma^{+}_{t}|= e^{t}|\Sigma^{+}_{0}|.
    \end{equation}
    In particular, if $\Sigma_{0}=\partial M$ is outer-minimizing then $|\Sigma_{t}| = e^{t} |\Sigma_{0}|$.
    \item (Lemma 7.1, Eccentricity) There exist constants $c,C>0$ so that for all $t$ sufficiently large

    \begin{equation} \label{eccen}
        ce^{\frac{t}{n-1}} \leq \min_{\Sigma_{t}} |x| \leq \max_{\Sigma_{t}} |x| \leq Ce^{\frac{t}{n-1}}
    \end{equation}
    in asymptotically flat coordinates $(x^{1},\dots,x^{n})$ of $(M^{n} \setminus K,g)$. 
    \item (Lemma 7.1 and Theorem 2.1, Asymptotitic Limit) The total weak mean curvature has the asymptotic limit
    \begin{equation} \label{total_mc}
        \lim_{t \rightarrow \infty} |\Sigma_{t}|^{-\frac{n-2}{n-1}} \int_{\Sigma_{t}} H d\sigma = (n-1)w_{n-1}^{\frac{1}{n-1}}.
    \end{equation}
\end{enumerate}
\end{theorem}

We call $t$ a ``jump time" if $\Sigma_{t} = \partial \{ u < t \} \neq \Sigma_{t}^{+}= \partial \{ u > t \}$. Since the surfaces $\Sigma^{+}_{t}$ are strictly outer-minimizing by property (1), we have at a jump time that the complement $\Sigma^{+}_{t} \setminus \Sigma_{t}$ is area-minimizing and hence has weak mean curvature $H_{\Sigma^{+}_{t} \setminus \Sigma_{t}}=0$ $\mathcal{H}^{n-1}$ a.e at these times. Combined with the area growth formula \eqref{area_growth} and the convergences \eqref{convergence}, this means that the quantity $Q(t)$ from \eqref{Q} strictly decreases across jumps and should therefore remain monotonic in $t$. To show this formally, we follow the work of Wei in \cite{W18}, c.f. also \cite{FS13}. We first need an inequality for the level sets of smooth functions in $(M^{n},g)$ (this is an exact identity in the vacuum case, see \cite{W18}, Lemma 4.1).
\begin{proposition}
Let $(M^{n},g,V)$ be asymptotically flat, and let $u \in C^{\infty}(M^{n})$ be a proper smooth function with $u|_{\partial M} = 0$. Call $\Omega_{t} = \{ u \leq t \}$, and let $\Phi: (0,t) \rightarrow \mathbb{R}$ be a compactly-supported Lipschitz function. Then under the TCC we have for the function $\phi(x) = (\Phi \circ u)(x)$ that

\begin{equation} \label{integral_identity}
    - \int_{\Omega_{t}} V \frac{\partial \phi}{\partial \nu} d\sigma \leq \int_{\Omega_{t}} \phi \left( 2 H \frac{\partial V}{\partial \nu} +VH^{2} - V|A|^{2} \right) d\Omega,
\end{equation}
where $\nu$, $A$, and $H$ are the outer normal, second fundamental form, and mean curvature of the regular level sets of $u$, respectively.
\end{proposition}
\begin{proof}
By the same choice of variation field as in \cite{W18}, Lemma 4.1 and the calculation therein, we obtain
\begin{eqnarray}
   \int_{\Sigma_{s}} \frac{V}{|\nabla u|} \frac{\partial H}{\partial \nu} d\sigma &=& -\int_{\Sigma_{s}} V \Delta_{\Sigma_{s}} |\nabla u|^{-1} - \int_{\Sigma_{s}} \frac{V}{|\nabla u|} \left( |A|^{2} + \text{Ric}(\nu,\nu) \right) d\sigma \\
   &=& -\int_{\Sigma_{s}} \frac{1}{|\nabla u|} \left( \Delta_{\Sigma_{s}} V + V|A|^{2} + V \text{Ric}(\nu,\nu) \right) d\sigma \nonumber \\
   &=& -\int_{\Sigma_{s}} \frac{1}{|\nabla u|} \left( \Delta V - \nabla^{2} V(\nu,\nu) + V \text{Ric}(\nu,\nu) \right) d\sigma \nonumber \\ 
   & & - \int_{\Sigma_{s}} \frac{1}{|\nabla u|} \left( V|A|^{2} - H \frac{\partial V}{\partial \nu} \right) d\sigma \nonumber \\
   &\leq& \int_{\Sigma_{s}} \frac{1}{|\nabla u|} \left( H\frac{\partial V}{\partial \nu} - V|A|^{2} \right) d\sigma \nonumber
\end{eqnarray} 
over the regular level sets $\Sigma_{s}= \{ u = s \}$ of $u$ (note that $\Sigma_{s}$ for a.e. $s \in (0,\infty)$ is regular by Sard's theorem). In the second line, we used the divergence theorem twice, and in the last line we once again used non-negativity of $P$. From here, repeating the remaining steps in \cite{W18}, Lemma 4.1 yields \eqref{integral_identity}.
\end{proof}

A weak solution $u \in C^{0,1}_{\text{loc}} (M^{n})$ to IMCF is realized as the uniform limit of proper smooth functions $u_{\epsilon}$ that each solve a regularized elliptic problem. In view of this and the weak convergence of corresponding weak mean curvatures $H_{\epsilon}$, \eqref{integral_identity} extends to the sub-level sets $\Omega_{t} = \{ u < t \}$ of $u$, where $H$ and $A$ are interpreted as the weak mean curvature and weak second fundamental form of $\Sigma_{t}=\partial \Omega_{t}$. From here, the same choice of Lipschitz function $\Phi$ as in \cite{W18}, Lemma 4.2, implies
\begin{corollary} \label{cor}
Let $ \{ \Sigma_{t} \}_{t \in (0,\infty)}$ be the proper weak IMCF of $\Sigma_{0}=\partial M$. Then under the TCC we have for $0 < t_{1} < t_{2}$ that

\begin{equation} \label{VHgrowth}
    \int_{\Sigma_{t_{2}}} VH d\sigma \leq \int_{\Sigma_{t_{1}}} VH d\sigma + \frac{n-2}{n-1} \int_{t_{1}}^{t_{2}} \left( \int_{\Sigma_{s}} VH + 2 \frac{n-1}{n-2} \frac{\partial V}{\partial \nu} d\sigma \right) ds,
\end{equation}
where $H$ is the weak mean curvature of $\Sigma_{t}$.
\end{corollary}
We are ready to prove the inequality for outer-minimizing boundaries. The case of multiple boundary components is handled identically to \cite{HI99}, Section 6 and \cite{W18}, Section 4.2, wherein at discrete times $\Sigma_{t}$ is replaced with the minimizing hull of $\Sigma_{t}$ and one of the horizons. Note that in dimensions greater than $7$ this hull may have a nonempty singular set, rendering it invalid initial data for weak IMCF. As a result, the inequality for multiple boundary components to dimension $3 \leq n \leq 7$.
\begin{proof}[Proof of Inequality \eqref{main_ineq}]
First, we consider the case where $\partial M^{n}=\Sigma$ is connected. Let $\{ \Sigma_{t} \}_{t \in (0,\infty)}$ be the proper weak IMCF of $\Sigma_{0}= \Sigma$. Applying the divergence theorem over $M^{n} \setminus \Omega_{t}$, we have that 

\begin{equation*}
    \int_{\Sigma_{t}} \frac{\partial V}{\partial \nu} d\sigma \leq (n-2)w_{n-1} m.
\end{equation*}
Define $q(t)= \int_{\Sigma_{t}} VH d\sigma + 2(n-1)w_{n-1} m$. By Corollary \ref{cor} and the above inequality, we have for any times $0 \leq t_{1} < t_{2}$ that

\begin{equation*}
    q(t_{2}) - q(t_{1}) \leq \frac{n-2}{n-1} \int_{t_{2}}^{t_{1}} q(s)ds.
\end{equation*}
Note also that the case $t_{1} =0$ follows from the $C^{1,\alpha}$ convergence $\Sigma_{t} \rightarrow \Sigma_{0}^{+}$ as $t \searrow 0$. Now let $Q(t)= |\Sigma_{t}|^{-\frac{n-2}{n-1}} q(t)$. Given $\Sigma_{0}$ outer-minimizing, it follows from the area growth formula \eqref{area_growth} for weak IMCF and Gronwall's Lemma that $Q(t_{2}) \leq Q(t_{1})$ for $t_{2} < t_{1}$.

We also address monotonicity for the case that $\partial M \setminus \Sigma= \cup_{i=1}^{I} \Sigma_{i}$ is non-empty. Once again by the divergence theorem and the fact that $V|_{\Sigma_{i}} \equiv 0$, the flux integral bound \eqref{imcf_mass} carries over to $\Sigma_{t}$. For $3 \leq n \leq 7$, there exists a flow $\{ \Sigma'_{t} \}_{t \in (0,\infty)}$ of $C^{1,\alpha}$ hypersurfaces in $(M^{n},g)$ such that $\Sigma'_{t}=\Sigma_{t}$ solves weak IMCF at all but finitely many times $T_{1}, \dots, T_{I-1}$ and, at these times, the surface $\Sigma_{T_{i}}$ of weak IMCF is replaced with the minimizing hull $\Sigma_{t_{i}}'$ of $\Sigma_{t_{i}} \cup \Sigma_{i}$, which is a $C^{1,\alpha}$ hypersurface and smooth away from the obstacle. Once again, the $C^{1,\alpha}$ convergences \eqref{convergence} imply 
\begin{equation*}
    \lim_{t \nearrow T_{i}} Q(t) \geq \lim_{t \searrow T_{i}} Q(t),
\end{equation*}
Therefore, repeating the argument for a connected boundary yields $Q(t_{2}) \leq Q(t_{1})$ in both cases. Finally, given the asymptotics \eqref{V_decay} for the potential and \eqref{total_mc} for total mean curvature, we get

\begin{equation*}
   \lim_{t \rightarrow \infty} Q(t) = \lim_{t \rightarrow \infty} |\Sigma_{t}|^{-\frac{n-2}{n-1}} \int_{\Sigma_{t}} H d\sigma = (n-1) w_{n-1}^{\frac{1}{n-1}}.
\end{equation*}
\end{proof}

\section{The Equality Case}
In what follows, we require the second Bianchi identity

\begin{equation} \label{bianchi}
    \text{div} \left( \text{Ric} \right) = \frac{1}{2} d R
\end{equation}
and the Ricci identity

\begin{equation} \label{ric}
    \text{div} \left( \nabla^{2} f \right) = \text{Ric} * \nabla f + d (\Delta f).
\end{equation}
for the divergence one-forms of the Ricci and Hessian tensors, respectively. In \eqref{ric}, ``$*$" denotes tensor contraction, that is $(T*X) (Y) = T(X, Y)$ for $X, Y \in \Gamma(TM)$.

\begin{theorem} \label{eq}
Let $(M^{n},g,V)$ be an asymptotically flat static triple satisfying the TCC. Suppose $(M^{n},g)$ is foliated by a smooth IMCF $X: \Sigma^{n-1} \times (0,\infty) \rightarrow (M^{n},g)$ such that $Q(t)=(n-1) w_{n-1}^{\frac{1}{n-1}}$ for each $t \in (0,\infty)$. Then the following hold:

\begin{enumerate}[label= (\roman*)]
    \item
    \begin{equation} \label{quasi_spherical}
        (M^{n},g) \cong (\mathbb{S}^{n-1} \times (r_{0},\infty), u^{2}(\theta,r) dr^{2} + r^{2} g_{S^{n-1}}),
    \end{equation}
    where $u$ is a smooth function on $\mathbb{S}^{n-1} \times (r_{0},\infty)$ and $g_{S^{n-1}}$ is the unit sphere metric.

    \item $(M^{n},g,V)$ is static vacuum; that is the metric and lapse solve the system

    \begin{eqnarray}
        \nabla^{2} V &=& V \text{Ric}, \label{vacuum} \\
        \Delta V &=& 0. \nonumber
    \end{eqnarray}
\end{enumerate}
\end{theorem}
\begin{proof}
Since $Q'(t)$ identically vanishes, we have from the evolution formula \eqref{ev_Q} that each leaf $\Sigma_{t}$ of IMCF is totally umbilical. Since $(M^{n},g)$ is asymptotically flat, \cite{HW24a} Theorem 3.1-- c.f. also \cite{HW24}, Theorem 3.2-- implies that $(M^{n},g)$ is of the form \eqref{quasi_spherical}. Here, the $\{ r= \text{const} \}$ slices in \eqref{quasi_spherical} are related to the leaves $\Sigma_{t}$ of IMCF by the change of variable

\begin{equation} \label{trans}
    r(t)= r_{0} e^{\frac{t}{n-1}}, \hspace{2cm} r_{0}= \left( \frac{|\Sigma_{0}|}{w_{n-1}} \right)^{\frac{1}{n-1}}.
\end{equation}
Next, the evolution formula \eqref{ev_Q} gives that

\begin{equation*}
    \int_{\{r=c\}} \frac{\partial V}{\partial \nu} d\sigma = (n-2) w_{n-1} m
\end{equation*}
for each $c > r_{0}$. Since $\Delta V \geq 0$, the above equality and the divergence theorem imply that $\Delta V = 0$ in $(M^{n},g)$. This also means that

\begin{equation} \label{trace}
    \text{Tr} \left( P \right)= V R + (n-1) \Delta V \equiv VR.
\end{equation}
Lastly, \eqref{ev_Q} gives that $P(\nu,\nu)$ vanishes on each leaf $\Sigma_{t}$, or equivalently

\begin{equation} \label{Prr}
    P_{rr} \equiv 0
\end{equation}
within the gauge \eqref{quasi_spherical}. Let $\{ \theta^{a} \}_{a=1}^{n-1}$ be normal coordinates about a point $\theta \in (\mathbb{S}^{n-1},g_{S^{n-1}})$. At $(\theta,r) \in \mathbb{S}^{n-1} \times (r_{0},\infty)$, we have for any constant $\alpha \in \mathbb{R}$ that
\begin{equation*}
    0 \leq P(\partial_{\theta^{a}} + \alpha \partial_{r}, \partial_{\theta^{a}} + \alpha \partial_{r}) = P_{aa} + 2 \alpha P_{ar} + \alpha^{2} P_{rr}
\end{equation*}
by the TCC. Therefore, $-2\alpha P_{ar} \leq P_{aa}$ for any $\alpha \in \mathbb{R}$. Taking $\alpha \rightarrow \pm \infty$, we get that

\begin{equation} \label{Paa}
    P_{ar} \equiv 0,
\end{equation}
that is, $\partial_{r}$ is an eigenvector of $P$. Now let us consider the divergence one-form $\text{div}\left( P \right)$. Combining \eqref{trace}-\eqref{Paa} and using the orthonormal frame $\{ u^{-1} \partial_{r}, r^{-1} \partial_{\theta^{1}}, \dots, r^{-1} \partial_{\theta^{n-1}} \}$ of $g$ at $(\theta,r)$, we may compute
\begin{eqnarray} \label{div2}
    \text{div} P (\partial_{r}) &=& u^{-1} \nabla_{\partial_{r}} P(u^{-1} \partial_{r},\partial_{r}) + r^{-2} \nabla_{\partial_{\theta^{a}}} P(\partial_{\theta^{a}},\partial_{r}) \nonumber \\
    &=& - u^{-3} P_{rr} \partial_{r} u + u^{-1} \left( \partial_{r} P_{rr} - 2P(\nabla_{\partial_{r}} \partial_{r}, \partial_{r}) \right) \\
    & & + r^{-2} \left( \partial_{a} P_{ar} - P(\nabla_{\partial_{\theta^{a}}} \partial_{\theta^{a}}, \partial_{r}) - P(\nabla_{\partial_{\theta^{a}}} \partial_{r}, \partial_{\theta^{a}}) \right) \nonumber \\
    &=& -r^{-3} P_{aa} = - r^{-1} \text{Tr}\left( P \right) = - V r^{-1} R, \nonumber
\end{eqnarray}
On the other hand, using the definition \eqref{P} for $P$, the second Bianchi identity \eqref{bianchi}, and the Ricci identity \eqref{ric}, we find that

\begin{eqnarray} \label{div1}
    \text{div} P &=& \text{div} \left( V \text{Ric} \right)  - \text{div} \nabla^{2} V + d (\Delta V) \nonumber \\
                     &=& \frac{V}{2} d R + \text{Ric} * \nabla V - \text{Ric} * \nabla V - d(\Delta V) + d(\Delta V) \nonumber \\
                     &=& \frac{V}{2} d R,
\end{eqnarray}
Comparing \eqref{div2} and \eqref{div1}, we obtain the following differential equation for the scalar curvature:
\begin{equation*}
    \partial_{r} R (r,\theta)= -2r^{-1} R(r,\theta).
\end{equation*}
We may integrate this expression in $r$ to obtain

\begin{equation} \label{R}
    R(\theta,r) = f(\theta) r^{-2}
\end{equation}
for some function $f \in C^{\infty} (\mathbb{S}^{n-1})$. By the eccentricity estimate in Theorem \ref{properties} (v), we have for $r$ in \eqref{trans} that $c |x| \leq r \leq C |x|$ in asymptotically flat coordinates $(x^{1}, \dots, x^{n})$ of $(M^{n} \setminus K, g)$. Since $g_{ij} - \delta_{ij} =o_{2}(1)$, the scalar curvature satisfies $R = o(r^{-2})$. This is only possible if $f \equiv 0$ in \eqref{R}, and so $(M^{n},g)$ is scalar-flat. 
Because $P$ is non-negative definite, $\text{Tr}(P) = VR \equiv 0$ implies that $P$ identically vanishes on $M^{n}$. 
\end{proof}

Riemannian manifolds of the form \eqref{quasi_spherical} are called quasi-spherical metrics with zero shear vector in the mathematical relativity literature, but for simplicity's sake we refer to these as \textit{quasi-spherical metrics}. In joint work \cite{HW24a} with Ye-Kai Wang, we applied foundational work by Bartnik \cite{B93} and Shi-Tam \cite{ST02} to obtain a complete characterization of quasi-spherical static vacuum metrics-- namely, all static vacuum metrics of the form \eqref{quasi_spherical} belong to the Schwarzschild family. As an immediate consequence, the triple $(M^{n},g,V)$ described in Theorem \ref{eq} is isometric to a piece of Schwarzschild. To accommodate for weak IMCF, we may use the following result from \cite{HW24a}, which generalizes earlier work by Huisken-Ilmanen \cite{HI08} and Li-Wei \cite{LW17}.

\begin{theorem}[\cite{HW24a}, Theorem 2.1] \label{smoothing}
Let $\{ \Sigma_{t} \}_{0 \leq t < \infty}$ be a proper weak IMCF within an asymptotically flat Riemannian manifold $(M^{n},g)$. Then there exists a $T_{0} < +\infty$ so that $\{ \Sigma_{t} \}_{t \in (T_{0},\infty)}$ is a smooth solution to IMCF \eqref{IMCF}.
\end{theorem}
The blow-down lemma (Lemma 7.1 of \cite{HI99}) of weak IMCF in asymptotically flat manifolds gives full $C^{1,\alpha}$ convergence of the rescaled $\Sigma_{t}$ to the unit sphere. So as in Theorem 2.7 of \cite{HI08}, $\Sigma_{t}$ are eventually $C^{1,\alpha}$ hypersurfaces, and thus there is no dimension restriction for Theorem \ref{smoothing}.

\begin{proof}[Proof of Theorem \ref{main}]
It remains to characterize the equality case. Let $K = \{ u \leq T_{0}\} \subset M^{n}$ where $u \in C^{0,1}_{\text{loc}} (M^{n})$ is the defining function associated with the weak IMCF of $\partial M$ and $T_{0}$ is the smoothing time in Theorem \ref{smoothing}. Then by Theorem \ref{eq}, $(M^{n} \setminus K,g)$ is quasi-spherical and static vacuum. Theorem 1.2 in \cite{HW24a} then gives that $(M^{n} \setminus K,g)$ is isometric to a rotationally symmetric region $\{ r \geq r_{0} \}$ of Schwarzschild space \eqref{schwarzschild}. In fact, if $T_{0} > 0$ then by the $C^{1,\alpha}$ convergence \eqref{convergence} we may find $\Sigma_{t}$ for $t \in (T_{0} - \epsilon,T_{0})$ satisfying the hypotheses of \cite{HW24a}, Theorem 2.2. Theorem 2.1 then implies that $\Sigma_{t}$ is smooth for $t < T_{0}$, a contradiction. Hence we conclude that $T_{0} = 0$, $(M^{n},g)$ is foliated by a smooth IMCF, and $(M^{n},g,V)$ is a rotationally symmetric piece of Schwarzschild.
\end{proof}

\appendix
\section{Proof of Proposition \ref{horizon}}
Unlike in \cite{HMM17}, we take the lapse to be strictly positive in this paper, thus simplifying some of the arguments.
\begin{proposition}
Let $(M^{n},g,V)$ be an asymptotically flat static triple with boundary $\partial M$ such that $P$ is non-negative definite in $(M^{n},g,V)$. Let $\Sigma_{i} \subset \partial M$ be a horizon. Then $V$ is either identically zero or is strictly positive on $\Sigma_{i}$.
\end{proposition}

\begin{proof}
Since $\Sigma_{i}$ is a stable minimal hypersurface, we have for any $\phi \in C^{\infty}(\Sigma_{i})$ that

\begin{equation} \label{variation}
    \int_{\Sigma_{i}} |\nabla \phi|^{2} d\sigma \geq \int_{\Sigma_{i}} \left( |A|^{2} + \text{Ric}(\nu,\nu) \right) \phi^{2} d\sigma
\end{equation}
by the stability inequality. Since $P(\nu,\nu) \geq 0$ on $\Sigma_{i}$ and $H_{\Sigma_{i}} \equiv 0$, at any $x \in \Sigma_{i}$ the Laplacian of $V$ satisfies

\begin{eqnarray*}
    \Delta V(x) &=& \Delta_{\Sigma_{i}} V + H \frac{\partial V}{\partial \nu} + \nabla^{2} V(\nu,\nu) \\
                       &\leq& \Delta_{\Sigma_{i}} V + V \text{Ric}(\nu,\nu) + \Delta V(x),
\end{eqnarray*}
and so $-\Delta_{\Sigma_{i}} V(x) \leq V \text{Ric}(\nu,\nu)$ on $\Sigma_{i}$. Multiplying this by $V$ and integrating then yields 

\begin{equation} \label{variation2}
   \int_{\Sigma_{i}} \text{Ric}(\nu,\nu) V^{2} d\sigma \geq  \int_{\Sigma_{i}} |\nabla V|^{2} d\sigma.
\end{equation}
Suppose that $V$ is not identically $0$ on $\Sigma_{i}$. Comparing \eqref{variation} and \eqref{variation2}, we see that $V$ is a first eigenfunction of the stability operator $S \phi = -\Delta_{\Sigma_{i}} \phi - \left(|A|^{2} + \text{Ric}(\nu,\nu) \right) \phi$ on $\Sigma_{i}$ with eigenvalue $0$. This implies that $V$ is non-vanishing and hence strictly positive on $\Sigma_{i}$.
\end{proof}
For non-vacuum manifolds, we cannot obtain the same local characterization of geometry near the horizon like in \cite{HMM17}, but we can at least obtain the same foliation by minimal hypersurfaces.
\begin{theorem}
Let $(M^{n},g,V)$ be an asymptotically flat static triple with boundary $\partial M$ such that the tensor $P$ in \eqref{P} is non-negative definite. Suppose $\Sigma_{i} \subset \partial M$ is a horizon. Then $V|_{\Sigma_{i}} \equiv 0$. 
\end{theorem}
Suppose $V|_{\Sigma_{i}}\neq 0$. Then $V$ is strictly positive on $\Sigma_{i}$ by the previous proposition. We consider a one-parameter family of immersions $X: \Sigma^{n-1} \times (0, \epsilon) \rightarrow (M^{n},g)$ with

\begin{equation} \label{V_flow}
    \frac{\partial}{\partial t} X(x,t) = V \nu(x,t) \hspace{2cm} (x,t) \in \Sigma^{n-1} \times (0, \epsilon).
\end{equation}
As in \cite{HMM17}, \cite{G93}, we obtain such immersions by applying the exponential map of the Fermat metric $\widetilde{g} = V^{-2} g$ to the outward unit normal field $\nu$ of $\Sigma_{i} \subset (M^{n},g)$. Using the variation formulas

\begin{eqnarray*}
    \frac{\partial}{\partial t } H &=& -\Delta_{\Sigma_{t}} \left(  V \right) - \left(|A|^{2} + \text{Ric}(\nu,\nu) \right) \left( V \right), \\
    \frac{\partial}{\partial t} V &=& V\frac{\partial V}{\partial \nu},
\end{eqnarray*}
we find along \eqref{V_flow} that

\begin{eqnarray*}
    \frac{\partial} {\partial t} \left( \frac{H}{V} \right) &=& -V^{-1} \left( \Delta_{\Sigma_{t}} V + \left( |A|^{2} + \text{Ric}(\nu,\nu) \right) V \right) - V^{-1} H \frac{\partial V}{\partial \nu} \\
    &=& -V^{-1} \left(\Delta V - \nabla^{2} V(\nu,\nu) + V\text{Ric}(\nu,\nu) + |A|^{2} V \right) \\
    &=& -V^{-1} \left( P(\nu,
    \nu) + |A|^{2} \right) \leq 0. 
\end{eqnarray*}
Therefore, $\Sigma_{t}=X\left( \Sigma \times \{ t \} \right)$ has mean curvature $H_{\Sigma_{t}} \leq 0$ in $(M^{n},g)$ for each $t \in (0,\epsilon)$. On the other hand, for any $t \in (0,\epsilon)$ we have the first variation of area that

\begin{equation}
    |\Sigma_{t}|-|\Sigma_{0}| = \int_{0}^{t} \int_{\Sigma_{s}} VH d\sigma ds \leq 0.
\end{equation}
This contradicts the locally area-minimizing property of $\Sigma_{0}$, and so we conclude that $V|_{\Sigma_{i}} \equiv 0$.
\printbibliography[title=References]
\end{document}